\documentclass[11pt]{article}

\usepackage{amsmath, amssymb, amsthm, graphicx, color, xspace}
\usepackage{hyperref}
\usepackage{algpseudocode}
\usepackage{mdframed}

\usepackage{float}
\newfloat{listing}{tbhp}{lst}
\floatname{listing}{Listing}

\newtheorem{thm}{Theorem}
\newtheorem{lem}[thm]{Lemma}

\theoremstyle{definition}
\newtheorem{dfn}[thm]{Definition}

\definecolor{codebgcolor}{RGB}{240,240,240}
\newcommand{\funcname}[1]{{\textsc{#1}}}

\title{A linear-time algorithm for the maximum-area inscribed triangle in a convex polygon}
\author{Yoav Kallus\footnote{Santa Fe Institute, 1399 Hyde Park Road, Santa Fe, NM 87501, USA. Email: \texttt{yoav@santafe.edu}}}
\date{June 8, 2017}

\begin{document}
\maketitle

\begin{abstract}
    Given the $n$ vertices of a convex polygon in cyclic order, can the triangle of maximum area inscribed in $P$ be determined by an
    algorithm with $O(n)$ time complexity? A purported linear-time algorithm by Dobkin and Snyder from 1979 has recently been shown
    to be incorrect by Keikha, L\"offler, Urhausen, and van der Hoog. These authors give an alternative algorithm with $O(n \log n)$ time
    complexity. Here we give an algorithm with linear time complexity.
\end{abstract}

\renewcommand{\thefootnote}{\fnsymbol{footnote}} 
\renewcommand{\thefootnote}{\arabic{footnote}} 

\section{Introduction}

The problem of finding a triangle inscribed in a convex polygon that maximizes the area among all inscribed triangles is a classical problem in computational geometry.
Recently, Keikha, L\"offler, Urhausen, and van der Hoog \cite{keikha} showed that an algorithm of Dobkin and Snyder \cite{dobkin} that purports to find the largest
triangle inscribed in a convex $n$-gon in $O(n)$ computational steps is incorrect by presenting a 9-gon for which the algorithm fails to find the largest triangle.
There are algorithms known that solve the problem in $O(n\log n)$ time complexity \cite{boyce,keikha}.
These algorithms make use of the subproblem of rooted triangles --- the maximum area triangle with some fixed vertex.
Here, by focusing on a different subproblem, that of anchored triangles --- the maximum area triangle with one side parallel to some fixed direction ---
we are able to produce an algorithm that finds the largest triangle inscribed in an $n$-gon in $O(n)$ computational steps.

\section{Theory}

We begin by defining the notion of a triangle inscribed in $P$ that is anchored to a direction $\mathbf{u}\in S^1$,
as previously introduced in \cite{kallus}.

\begin{dfn}
    Let $P$ be a convex polygon, $\mathbf{u}\in S^2$ be a unit vector, and $\mathbf{a}\mathbf{b}\mathbf{c}$ a triangle inscribed in $P$.
    We say $\mathbf{a}\mathbf{b}\mathbf{c}$ is
    \textit{anchored} to $\mathbf{u}$ if $\mathbf{u}$ is the outer normal to the edge $\mathbf{b}\mathbf{c}$ and $\mathbf{a}\mathbf{b}\mathbf{c}$
    achieves the maximum area out of all triangles inscribed in $P$ with $\mathbf{u}$ as the outer normal to the edge $\mathbf{b}\mathbf{c}$.
\end{dfn}

Implicit in the definition of a triangle anchored to $\mathbf{u}$ are the identities of its vertices $\mathbf{a}$, $\mathbf{b}$, and $\mathbf{c}$ which
are always to be in counter-clockwise order, with the latter two delimiting the edge to which $\mathbf{u}$ is normal. We will also say a triangle is \textit{anchored}
if it is anchored to some direction.

Clearly, the maximum-area triangle inscribed in $P$ is anchored to some $\mathbf{u}\in S^1$. In fact, it is anchored to all three of its edge outer normals.
Therefore, all three vertices of the maximum-area triangle inscribed in $P$ are vertices of $P$.

\begin{figure}
    \center
    \includegraphics[width=0.49\textwidth]{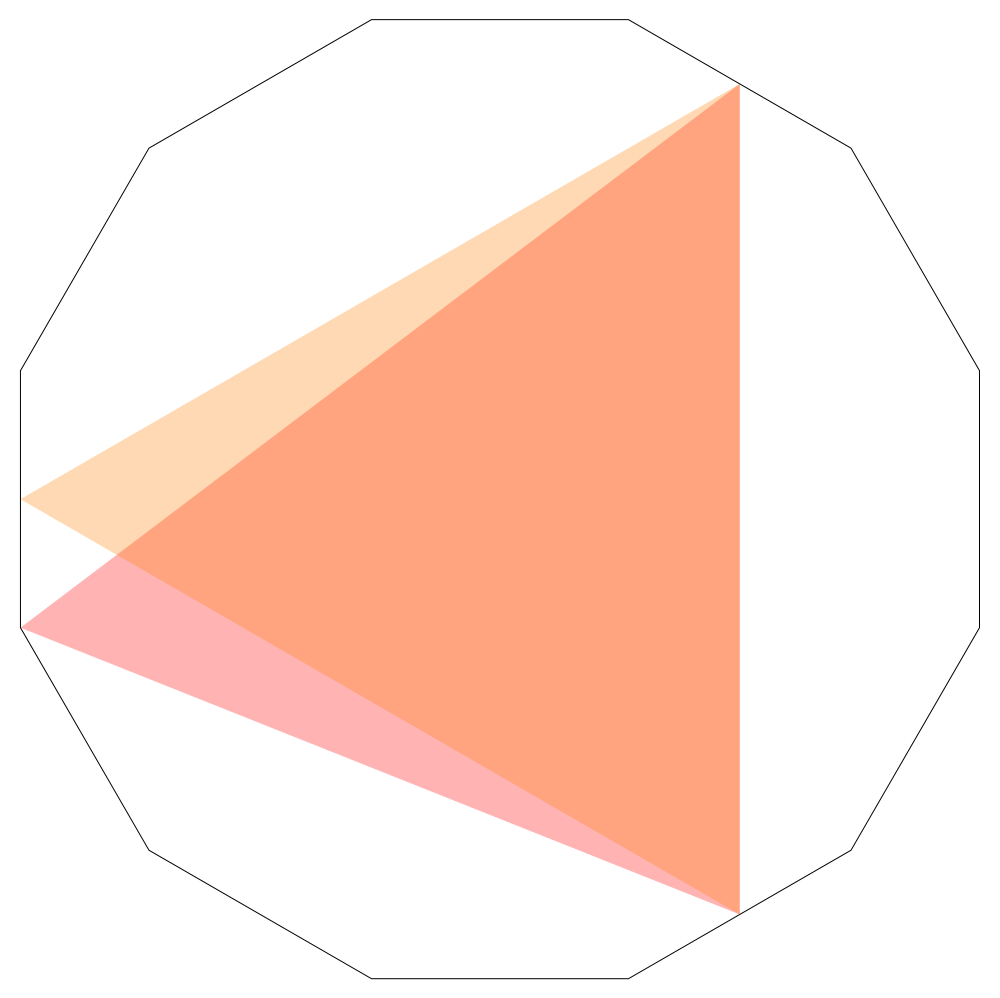}
    \includegraphics[width=0.49\textwidth]{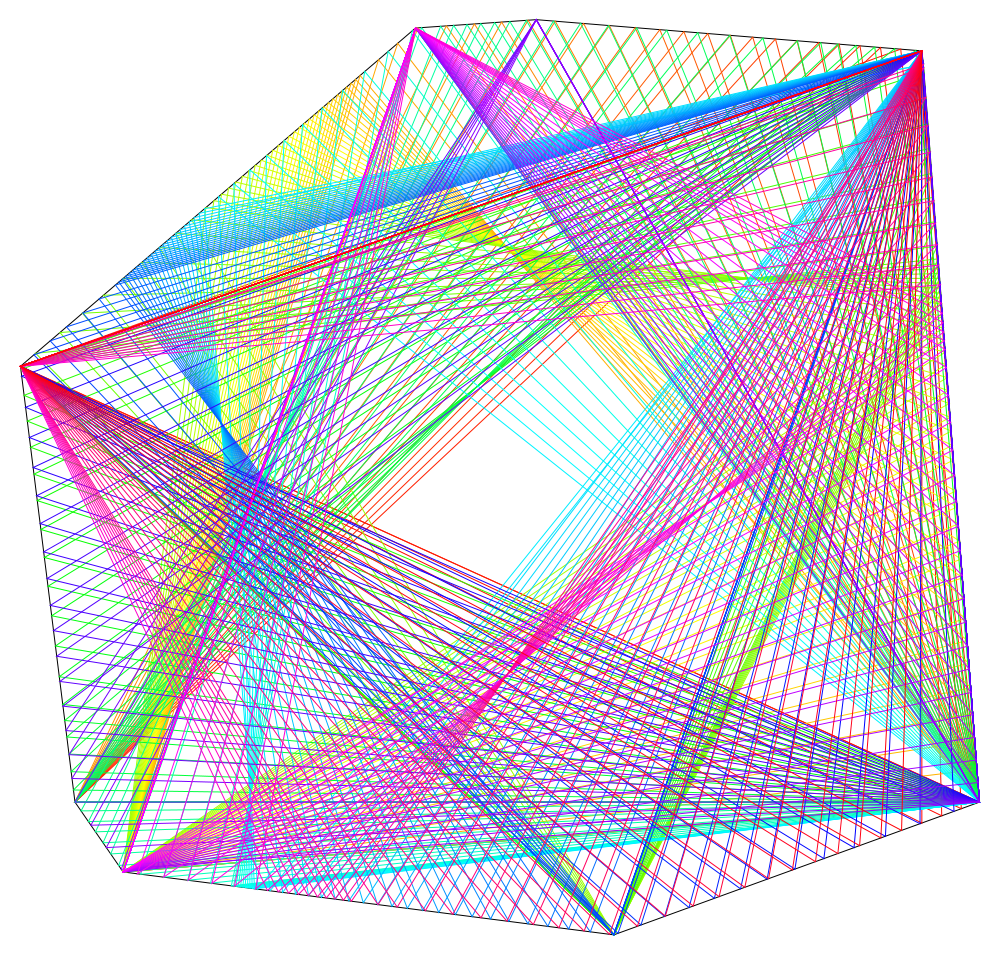}
    \caption{Left: two triangles inscribed in the regular dodecagon anchored to one of the dodecagon's edge normals.
	Right: triangles inscribed in the 9-gon from \cite{keikha} anchored to 240 equally spaced directions.}
\end{figure}

\begin{lem}
    \label{lem:uniq}
    \begin{enumerate}
	\item If $\mathbf{a}\mathbf{b}\mathbf{c}$ is anchored to $\mathbf{u}$, then $-\mathbf{u}$ is an outer normal to $P$ at $\mathbf{a}$.
	\item If $\mathbf{a}\mathbf{b}\mathbf{c}$ and $\mathbf{a}'\mathbf{b}'\mathbf{c}'$ are both anchored to $\mathbf{u}$, then $\mathbf{b}=\mathbf{b}'$ and $\mathbf{c}=\mathbf{c}'$.
	\item If $\mathbf{a}\mathbf{b}\mathbf{c}$ is a local maximum of the area among triangles inscribed in $P$ with $\mathbf{u}$ as the outer normal to the edge $\mathbf{b}\mathbf{c}$, then it is also a global maximum.
    \end{enumerate}
\end{lem}

\begin{proof}
    To prove (1), suppose that $-\mathbf{u}$ is not an outer normal to $P$ at $\mathbf{a}$, then pick some $\mathbf{a}'$ such
    that $-\mathbf{u}$ is an outer normal to $P$ at $\mathbf{a}'$. Then $-\mathbf{u}\cdot\mathbf{a}' > -\mathbf{u}\cdot\mathbf{a}$,
    and therefore $\mathbf{a}'\mathbf{b}\mathbf{c}$ has a larger area than $\mathbf{a}\mathbf{b}\mathbf{c}$, so the latter cannot
    be anchored at $\mathbf{u}$.

    Let $y_\text{min} = \min\{\mathbf{x}\cdot\mathbf{u}:\mathbf{x}\in P\}$ and $y_\text{max} = \max\{\mathbf{x}\cdot\mathbf{u}:\mathbf{x}\in P\}$.
    Now, let $\mathbf{v}=(-u_2,u_1)$, where $\mathbf{u}=(u_1,u_2)$ and define the functions
    $f(y) = \max\{z: y\mathbf{u}+z\mathbf{v}\in P\}$ and $g(y) = -\min\{z: y\mathbf{u}+z\mathbf{v}\in P\}$ for $y\in[y_\text{min},y_\text{max}]$.
    If $\mathbf{a}\mathbf{b}\mathbf{c}$ is inscribed in $P$ with $\mathbf{u}$ as the outer normal to
    the edge $\mathbf{b}\mathbf{c}$, then $\mathbf{b}=y\mathbf{u}-g(y)\mathbf{v}$ and $\mathbf{c}=y\mathbf{u}+f(y)\mathbf{v}$ for some $y$
    or $\mathbf{b}=y_\text{max}\mathbf{u}-z'\mathbf{v}$ and $\mathbf{c}=y_\text{min}\mathbf{u}+z''\mathbf{v}$ with $z'<f(y_\text{max})$
    or $z''<g(y_\text{max})$.
    If the latter case holds, than $\mathbf{a}\mathbf{b}\mathbf{c}$ cannot be of maximal area, even locally.
    The functions $f(y)$ and $g(y)$ (and so also $f(y)+g(y)$) are convex and piecewise linear on $(y_\text{min},y_\text{max})$.
    It follows that $s(y)=\sqrt{\tfrac12(y-y_\text{min})[f(y)+g(y)]}$ is also convex and has a single maximum.
    Therefore, if $\mathbf{a}\mathbf{b}\mathbf{c}$ is of maximal area among triangles with $\mathbf{u}$ as the outer normal, then
    $\mathbf{b}=y^*\mathbf{u}-g(y^*)\mathbf{v}$ and $\mathbf{c}=y^*\mathbf{u}+f(y^*)\mathbf{v}$, where $y^*$ is the unique
    maximum of $s(y)$. Furthermore, there are no local maxima except for ones with those vertices $\mathbf{b}$ and $\mathbf{c}$ that
    achieve the global maximum.
\end{proof}

When $-\mathbf{u}$ is an outer normal of $P$ at a unique point (necessary a vertex of $P$), then the triangle anchored to $\mathbf{u}$ is unique.
When $-\mathbf{u}$ is an outer normal to $P$ along an edge, we will choose among the anchored triangles the one such the $\mathbf{a}$ is the vertex
at the counter-clockwise end of that edge. Therefore, we will always have a unique triangle in mind when we refer to \textit{the} triangle anchored to $\mathbf{u}$.
We denote the vertices of that triangle as functions of $\mathbf{u}$: $\mathbf{a}(\mathbf{u})$, $\mathbf{b}(\mathbf{u})$, and $\mathbf{c}(\mathbf{u})$.

\begin{lem}
    \begin{enumerate}
	\item $\mathbf{a}: S^1 \to \partial P$ is piecewise constant.
	\item $\mathbf{b}, \mathbf{c}: S^1 \to \partial P$ are continuous.
    \end{enumerate}
\end{lem}
\begin{proof}
    \begin{enumerate}
	\item When $-\mathbf{u}$ is not an outer normal to $P$ along an edge, then any vector in a neighbor of $-\mathbf{u}$ is also normal to $P$ at $\mathbf{a}(\mathbf{u})$.
    Therefore, $\mathbf{a}(\mathbf{u})$ is locally constant. Since $P$ has finitely many edges, $\mathbf{a}(\mathbf{u})$ is locally constant
    at all but finitely many points.
	\item If $\mathbf{b}(\mathbf{u}_k)\to\mathbf{b}'$ and $\mathbf{c}(\mathbf{u}_k)\to\mathbf{c}'$ as $\mathbf{u}_k\to\mathbf{u}$, where
    $\mathrm{area}(\mathbf{a}\mathbf{b}'\mathbf{c}') < \mathrm{area}(\mathbf{a}(\mathbf{u})\mathbf{b}(\mathbf{u})\mathbf{c}(\mathbf{u}))$, then for sufficiently 
    large $k$, we could deform $\mathbf{a}(\mathbf{u})\mathbf{b}(\mathbf{u})\mathbf{c}(\mathbf{u})$ to create a triangle with $\mathbf{u}_k$ as an outer normal
    and area larger than $\mathbf{a}_k\mathbf{b}_k\mathbf{c}_k$. This is a contradiction, so
    $\mathrm{area}(\mathbf{a}\mathbf{b}'\mathbf{c}')\ge\mathrm{area}(\mathbf{a}\mathbf{b}'\mathbf{c}')$, and by uniqueness of the maximum,
    the functions are continuous.
    \end{enumerate}
\end{proof}

So, the anchored triangles of $P$ form a cyclic one-parameter family, in which two vertices move continuously and the third jumps
from one vertex to the next in cyclic order. If we could go through the anchored triangles in a systematic way, we could find the one that
maximizes the area, and we would be done. In the next section, we will show how this can be done, but we require some more preparation.

\begin{dfn}
    Let $P$ be a convex polygon, $\mathbf{u}\in S^1$ be a unit vector, and $\mathbf{a}\mathbf{b}\mathbf{c}$ a triangle inscribed in $P$.
    We say $\mathbf{a}\mathbf{b}\mathbf{c}$ is \textit{candidate-anchored} to $\mathbf{u}$ if it satisfies the conditions
    \begin{enumerate}
	\item $\mathbf{u}$ is an outer normal to the edge $\mathbf{b}\mathbf{c}$.
	\item $-\mathbf{u}$ is an outer normal to $P$ at $\mathbf{a}$.
    \end{enumerate}
\end{dfn}

For a point $\mathbf{x}\in\partial P$, define its forward counter-clockwise tangent $e_+(\mathbf{x})\in S^1$ to be
the direction from $\mathbf{x}$ to the vertex immediately counter-clockwise to $\mathbf{x}$.
Similarly, let $e_-(\mathbf{x})$, the backward counter-clockwise tangent, be the direction to $\mathbf{x}$ from the vertex immediately clockwise to $\mathbf{x}$.
The two are equal except at vertices of $P$.
Let $L_\pm(\mathbf{x})$ be the line $\{\mathbf{x}+s \mathbf{e}_\pm(\mathbf{x}): s\in\mathbb{R}\}$.
Let $\mathbf{a}\mathbf{b}\mathbf{c}$ be candidate-anchored to $\mathbf{u}$, and let $M$ be the line through $\mathbf{b}$ and $\mathbf{c}$.
Let $\mathbf{b}'_\pm(h)$ and $\mathbf{c}'_\pm(h)$ be the intersections of $M+\|\mathbf{c}-\mathbf{b}\|h\mathbf{u}$ with $L_\pm(\mathbf{b})$
and with $L_\pm(\mathbf{c})$, respectively. Then define 
$$
Q_{\pm\pm}(\mathbf{a},\mathbf{b},\mathbf{c}) = \left.\frac{d(\mathrm{area}(\mathbf{a}\mathbf{b}'_\pm(h)\mathbf{c}'_\pm(h))}{dh}\right|_{h=0}\text.
$$
Note that the intersection $\mathbf{b}'_\pm(h)$ does not exists if $e_\pm(\mathbf{b})$ is antiparallel or directly parallel to $\mathbf{c}-\mathbf{b}$.
Since the antiparallel case is impossible a limiting procedure suggests that we should take $Q_{\pm\pm}(\mathbf{a},\mathbf{b},\mathbf{c}) = -\infty$ in this case.
Similarly, if $e_\pm(\mathbf{c})$ is parallel to $\mathbf{c}-\mathbf{b}$ (again necessarily directly parallel),
we take $Q_{\pm\pm}(\mathbf{a},\mathbf{b},\mathbf{c}) = +\infty$. If both are true, then we leave $Q_{\pm\pm}$ undefined.

\begin{lem}
    Let $\mathbf{a}\mathbf{b}\mathbf{c}$ be candidate-anchored to $\mathbf{u}$.
    \begin{enumerate}
	\item $\mathbf{a}\mathbf{b}\mathbf{c}$ is anchored to $\mathbf{u}$ if and only if
	    $Q_{-+}(\mathbf{a},\mathbf{b},\mathbf{c})\ge 0$ and $Q_{+-}(\mathbf{a},\mathbf{b},\mathbf{c})\le0$.
	\item If neither $\mathbf{b}$ or $\mathbf{c}$ are vertices of $P$, then $\mathbf{a}\mathbf{b}\mathbf{c}$ is anchored
	    to $\mathbf{u}$ if and only if $Q_{++}(\mathbf{a},\mathbf{b},\mathbf{c})=0$.
	\item If $Q_{++}(\mathbf{a},\mathbf{b},\mathbf{c})=0$, then $\mathbf{a}\mathbf{b}\mathbf{c}$ is anchored to $\mathbf{u}$.
    \end{enumerate}
\end{lem}
\begin{proof}
    For (1), recall the functions $f(y)$, $g(y)$, and $s(y)$ from the proof of Lemma \ref{lem:uniq}.
    If $\mathbf{b}=y\mathbf{u}-g(y)\mathbf{v}$ and $\mathbf{c}=y\mathbf{u}+f(y)\mathbf{v}$, then
    $Q_{-+}(\mathbf{a},\mathbf{b},\mathbf{c}) = (1/\|\mathbf{c}-\mathbf{b}\|)d(s^2)/dy|_{y^-}$ and
    $Q_{+-}(\mathbf{a},\mathbf{b},\mathbf{c}) = (1/\|\mathbf{c}-\mathbf{b}\|)d(s^2)/dy|_{y^+}$ are
    the directional derivatives of the area $s(y)^2$.
    The triangle is anchored to $\mathbf{u}$ if and only if $y$ is a maximum
    of $s(y)$, which completes the proof.

    (2) is a simple corollary of (1), since $Q_{++}(\mathbf{a},\mathbf{b},\mathbf{c})=Q_{-+}(\mathbf{a},\mathbf{b},\mathbf{c})=Q_{+-}(\mathbf{a},\mathbf{b},\mathbf{c})=0$
    when neither $\mathbf{b}$ or $\mathbf{c}$ are vertices of $P$.

    Finally, because the tangent discontinuities at vertices of $P$ can only be in the convex direction,
    $\mathrm{area}(\mathbf{a}\mathbf{b}'_-(h)\mathbf{c}'_+(h))\le\mathrm{area}(\mathbf{a}\mathbf{b}'_+(h)\mathbf{c}'_+(h))$ when $h<0$ and
    $\mathrm{area}(\mathbf{a}\mathbf{b}'_+(h)\mathbf{c}'_-(h))\le\mathrm{area}(\mathbf{a}\mathbf{b}'_+(h)\mathbf{c}'_+(h))$ when $h>0$.
    Therefore, if $Q_{++}(\mathbf{a},\mathbf{b},\mathbf{c})=0$, then $Q_{-+}(\mathbf{a},\mathbf{b},\mathbf{c})\ge 0$
    and $Q_{+-}(\mathbf{a},\mathbf{b},\mathbf{c})\le0$, and by (1), the triangle is anchored.
\end{proof}

\begin{thm}
    \label{thm:forward}
    Let $\mathbf{a}\mathbf{b}\mathbf{c}$ be anchored to $\mathbf{u}=(\cos\theta,\sin\theta)$ and suppose
    that if $-\mathbf{u}$ is an outer normal to an edge of $P$, that $\mathbf{a}$ is at the counter-clockwise end of that edge.
    There exists $\delta>0$ such that when $\mathbf{u}'= (\cos\theta',\sin\theta')$ and $\theta'\in[\theta,\theta+\delta]$,
    there is a triangle $\mathbf{a}'\mathbf{b}'\mathbf{c}'$ anchored to $\mathbf{u}'$ such that $\mathbf{a}'=\mathbf{a}$,
    $\mathbf{b}'$ is on the edge segment between $\mathbf{b}$ and the vertex immediately counter-clockwise to $\mathbf{b}$, and
    $\mathbf{c}'$ is on the edge segment between $\mathbf{c}$ and the vertex immediately counter-clockwise to $\mathbf{c}$.
\end{thm}

\begin{proof}
    We will analyze three cases, $Q_{++}(\mathbf{a},\mathbf{b},\mathbf{c})=0$, $Q_{++}(\mathbf{a},\mathbf{b},\mathbf{c})<0$, and $Q_{++}(\mathbf{a},\mathbf{b},\mathbf{c})>0$
    (since $\mathbf{a}\mathbf{b}\mathbf{c}$ is anchored, $\mathbf{e}_+(\mathbf{c})$ is not parallel to $\mathbf{c}-\mathbf{b}$ so $Q_{++}(\mathbf{a},\mathbf{b},\mathbf{c})$
    is not undefined).

    In the first case, $Q_{++}(\mathbf{a},\mathbf{b},\mathbf{c})=0$,
    let $\mathbf{b}'(t_b)=\mathbf{b}+t_b\mathbf{e}_+(\mathbf{b})$,
    let $\mathbf{c}'(t_c)=\mathbf{c}+t_c\mathbf{e}_+(\mathbf{c})$,
    and consider the triangle $\mathbf{a}\mathbf{b}'(t_b)\mathbf{c}'(t_c)$.
    Let $\mathbf{u}'(t_b,t_c)$ be perpendicular to $\mathbf{c}'(t_c)-\mathbf{b}'(t_b)$,
    so the triangle is candidate-anchored to $\mathbf{u}'(t_b,t_v)$.
    For small enough nonnegative $t_b$ and $t_c$, $\mathbf{e}_+(\mathbf{b}'(t_b))=\mathbf{e}_+(\mathbf{b})$
    and $\mathbf{e}_+(\mathbf{c}'(t_c))=\mathbf{e}_+(\mathbf{c})$, and
    the equation $Q_{++}(\mathbf{a},\mathbf{b}'(t_b),\mathbf{c}'(t_c))=0$ gives an implicit relation
    between $t_b$ and $t_c$ of the form $\beta t_b - \gamma t_c + \alpha t_b t_c=0$,
    where $\beta=\mathbf{e}_+(\mathbf{b})\wedge(\mathbf{c}-\tfrac12\mathbf{a}-\tfrac12\mathbf{b})$,
    $\gamma=\mathbf{e}_+(\mathbf{c})\wedge(\mathbf{b}-\tfrac12\mathbf{a}-\tfrac12\mathbf{c})$,
    $\alpha=\mathbf{e}_+(\mathbf{b})\wedge\mathbf{e}_+(\mathbf{c})$, and
    $\wedge$ denotes the wedge product $(x_1,y_1)\wedge(x_2,y_2) = x_1y_2 - y_1x_2$.
    We now claim that $\beta\ge0$, $\gamma\ge0$, and at least one of them is positive.
    To show this most easily, we translate and rotate $P$ so that $\mathbf{a}=(0,0)$ and $\mathbf{u}=(1,0)$.
    Since $Q_{++}(\mathbf{a},\mathbf{b},\mathbf{c})=0$ we know neither $L_+(\mathbf{b})$ or $L_+(\mathbf{c})$ are
    vertical, so we can parameterize them as $\{(x,r_b + m_b x)\}$ and $\{(x,r_c + m_c x)\}$ respectively.
    The equations $(\mathbf{c}-\mathbf{b})\cdot\mathbf{u}=0$ and $Q_{++}(\mathbf{a},\mathbf{b},\mathbf{c})=0$
    together with $\mathbf{b}\in L_+(\mathbf{b})$ and $\mathbf{c}\in L_+(\mathbf{c})$, determine
    $\mathbf{b}$ and $\mathbf{c}$ to have $x=(r_c-r_b)/2(m_b-m_c)$. We now immediately get
    $\beta = r_c$ and $\gamma = -r_b$. Since $\mathbf{a}=(0,0)$ is contained in $P$, the support lines,
    $L_+(\mathbf{b})$ and $L_+(\mathbf{c})$ must intersect the $y$-axis below and above it respectively,
    so $r_b\le0$ and $r_c\ge0$. If both lines intersect at $\mathbf{a}$, then we have $x=0$ for $\mathbf{b}$ and
    $\mathbf{c}$, which is impossible. 
    Therefore, we can invert the implicit relations
    $\beta t_b - \gamma t_c + \alpha t_b t_c=0$
    and $(\mathbf{c}'(t_c)-\mathbf{b}'(t_b))\cdot\mathbf{u}'=0$ to obtain functions
    $t_b(\mathbf{u}')$ and $t_c(\mathbf{u}')$ defined in the neighborhood of $\mathbf{u}$ and in particular,
    for $\mathbf{u}'= (\cos\theta',\sin\theta')$, $\theta'\in[\theta,\theta+\delta]$.
    The triangles $\mathbf{a}\mathbf{b}'(t_b(\mathbf{u}'))\mathbf{c}'(t_c(\mathbf{u}'))$ have $Q_{++}=0$,
    and are therefore anchored to $\mathbf{u}'$.

    In the second case, $Q_{++}(\mathbf{a},\mathbf{b},\mathbf{c})<0$.
    Since $\mathbf{a}\mathbf{b}\mathbf{c}$ is anchored, we know that $\mathbf{b}$ is a vertex
    (or else $Q_{-+}(\mathbf{a},\mathbf{b},\mathbf{c})=Q_{++}(\mathbf{a},\mathbf{b},\mathbf{c})<0$).
    Let $\mathbf{c}'(t_c)=\mathbf{c}+t_c\mathbf{e}_+(\mathbf{c})$.
    Since $Q_{++}(\mathbf{a},\mathbf{b},\mathbf{c}'(t_c))$ evolves continuously and $Q_{++}(\mathbf{a},\mathbf{b},\mathbf{c}'(0))$ is negative,
    we have for small enough positive $t_c$ that $Q_{+-}(\mathbf{a},\mathbf{b},\mathbf{c}'(t_c)) = Q_{+-}(\mathbf{a},\mathbf{b},\mathbf{c}'(t_c))$ is also
    negative.
    If $Q_{-+}(\mathbf{a},\mathbf{b},\mathbf{c}'(0))>0$, then the same argument gives $Q_{-+}(\mathbf{a},\mathbf{b},\mathbf{c}'(t_c))>0$ for small enough $t_c$,
    and the triangle $\mathbf{a}\mathbf{b}\mathbf{c}'(t_c)$ is anchored.
    Therefore, the only problematic case is when $Q_{-+}(\mathbf{a},\mathbf{b},\mathbf{c}'(0))=0$.
    Again, for ease of analysis, we translate and rotate $P$ so that $\mathbf{a}=(0,0)$ and $\mathbf{u}=(1,0)$.
    Since $Q_{-+}(\mathbf{a},\mathbf{b},\mathbf{c})=0$ we know neither $L_-(\mathbf{b})$ or $L_+(\mathbf{c})$ are
    vertical, so we can parameterize them as $\{(x,r_b + m_b x)\}$ and $\{(x,r_c + m_c x)\}$ respectively,
    and we can solve for $\mathbf{b}$ and $\mathbf{c}$ on those lines from $Q_{-+}(\mathbf{a},\mathbf{b},\mathbf{c})=0$ and $(\mathbf{c}-\mathbf{b})\cdot\mathbf{u}=0$.
    We find that $dQ_{-+}/dt_c = 4(m_b-m_c)(-r_b)/(r_b-r_c)^2>0$.
    Therefore, $Q_{-+}(\mathbf{a},\mathbf{b},\mathbf{c}'(t_c))\ge0$ for small enough $t_c$ in this case too.
    If we let $t_c(\mathbf{u}')$ be the solution of $(\mathbf{c}'(t_c)-\mathbf{b})\cdot\mathbf{u}'=0$, we obtain the triangle
    $\mathbf{a}\mathbf{b}\mathbf{c}'(t_c(\mathbf{u}'))$, which we have shown to be anchored to $\mathbf{u}'= (\cos\theta',\sin\theta')$,
    when $\theta'\in[\theta,\theta+\delta]$ and $\delta$ is sufficiently small.

    The final case, $Q_{++}(\mathbf{a},\mathbf{b},\mathbf{c})>0$ is symmetric to the previous case:
    necessarily $\mathbf{c}$ is a vertex, and we consider the triangle $\mathbf{a}\mathbf{b}'(t_b(\mathbf{u}'))\mathbf{c}$,
    where $\mathbf{b}'(t_b)=\mathbf{b}+t_b\mathbf{e}_+(\mathbf{b})$. By a similar analysis as in the previous case,
    we conclude that this triangle is anchored for small enough positive $t_b$.
\end{proof}

The observation that the vertices of the triangle anchored to $\mathbf{u}$ move monotonically in a counter-clockwise way as $\mathbf{u}$
moves in a counter-clockwise way is crucial in demonstrating that the algorithm presented in the next section takes time linear in the number of vertices.

\section{Algorithmic implementation}

We first show that given a vector $\mathbf{u}$ and a polygon $P$ with $n$ vertices, it is possible to obtain
a triangle inscribed in $P$ anchored to $\mathbf{u}$ in linear time.
Start by setting $\mathbf{a}$ to the vertex $\mathbf{p}$ of $P$ minimizing $\mathbf{u}\cdot\mathbf{p}$,
and set both $\mathbf{b}$ and $\mathbf{c}$ to the vertex maximizing $\mathbf{u}\cdot\mathbf{p}$.
We evolve $\mathbf{b}$ to $\mathbf{b}'=\mathbf{b}-t_b\mathbf{e}_-(\mathbf{b})$ and $\mathbf{c}$ to $\mathbf{c}'=\mathbf{c}+t_c\mathbf{e}_+(\mathbf{c})$,
maintaining $(\mathbf{c}'-\mathbf{b}')\cdot\mathbf{u} = 0$ and stopping when the first of the following conditions arises:
\begin{enumerate}
    \item $\mathbf{b}'$ reaches the vertex immediately clockwise to $\mathbf{b}$,
    \item $\mathbf{c}'$ reaches the vertex immediately counter-clockwise to $\mathbf{c}$, or
    \item $Q_{-+}(\mathbf{a},\mathbf{b}',\mathbf{c}')$ reaches $0$.
\end{enumerate}
We repeat this iterative evolution until we have $Q_{-+}(\mathbf{a},\mathbf{b},\mathbf{c})\ge0$ and $Q_{+-}(\mathbf{a},\mathbf{b},\mathbf{c})\ge0$.

To find the maximal inscribed triangle in $P$, we start with a triangle anchored to some arbitrary $\mathbf{u}$.
We then determine which of the cases of the proof of Theorem \ref{thm:forward}
the triangle belongs to (shift $\mathbf{a}$ to the next vertex in the counterclockwise direction if needed to satisfy the hypothesis of the theorem),
and obtain the formula for evolving $\mathbf{b}'=\mathbf{b}+t_b\mathbf{e}_+(\mathbf{b})$ and
$\mathbf{c}'=\mathbf{c}+t_c\mathbf{e}_+(\mathbf{c})$ such that the triangle $\mathbf{a}\mathbf{b}'\mathbf{c}'$ is anchored;
we evolve forward to the point where the first of the following stopping conditions arises:
\begin{enumerate}
    \item $\mathbf{b}'$ reaches the vertex immediately counter-clockwise to $\mathbf{b}$,
    \item $\mathbf{c}'$ reaches the vertex immediately counter-clockwise to $\mathbf{c}$,
    \item $\mathbf{c}'-\mathbf{b}'$ becomes parallel to $\mathbf{e}_+(\mathbf{a})$, or
    \item $Q_{++}(\mathbf{a},\mathbf{b}',\mathbf{c}')$ reaches $0$ when it was previously not $0$.
\end{enumerate}
We continue this iterative evolution until we have evolved through all the anchored triangles of $P$.
Whenever the vertices of the triangle all coincide with vertices of $P$, we record the area of the triangle, and
store for output the triangle with the largest area.

We give a pseudocode representation of the algorithm for finding an anchored triangle in Listing \ref{alg:anchored} and for finding the largest area
inscribe triangle in Listing \ref{alg:max_inscribed}. Some cumbersome formulas were moved to Listing \ref{alg:util} for neatness.
A C++ implementation of the algorithm is available alongside the source code of this eprint, and also available at \url{https://github.com/ykallus/max-triangle/releases/tag/v1.0}.

\begin{thm} On the input $P=(\mathbf{p}_1,\ldots,\mathbf{p}_n)$ representing a convex $n$-gon,
    the algorithm represented in Listing \ref{alg:max_inscribed}, outputs the vertex indices of a triangle inscribed in $P$ that
    has the maximal area among all inscribed triangles. The algorithm halts after $O(n)$ iterations.
\end{thm}
\begin{proof}
    The triangle represented by the algorithm evolves continuously through all the anchored triangles of $P$, and any new vertex reached
    triggers a new iteration of the main loop. Therefore, the maximum area inscribed triangle is encountered by the algorithm at the
    top of the loop on some iteration, and its vertex indices will be recorded and returned.

    The initial task of finding an anchored triangle takes $O(n)$ iterations because at each iteration of \textsc{anchored\_triangle},
    either $\mathbf{b}$ or $\mathbf{c}$ visits a new vertex of $P$, or the algorithm halts. Since $\mathbf{b}$ and $\mathbf{c}$ never
    visit the same vertex more than once, this task is completed in $O(n)$ iterations.

    To see that \textsc{largest\_inscribed\_triangle} goes through $O(n)$ iterations, note that on each iteration either one of the three triangle vertices
    reaches some new vertex of $P$ or $Q$ switches from being zero to being nonzero. Since
    $\mathbf{a}$ visits all $n$ vertices in order, returning to the one it started on, and since $\mathbf{b}$ and $\mathbf{c}$ only
    move counter-clockwise and never cross each other or $\mathbf{a}$, it follows that each of them can reach a new vertex at most
    $2n$ times. Since $Q$ cannot switch from being zero to being nonzero on two consecutive iterations,
    the algorithm halts after $O(n)$ iterations.
\end{proof}

\begin{listing}
\begin{mdframed}[backgroundcolor=codebgcolor] 
\begin{algorithmic}
    \Function{anchored\_triangle}{$(\mathbf{p}_1,\ldots,\mathbf{p}_n)$,$\mathbf{u}$}
    \State $i_a\gets\arg\min_{i=1,\ldots,n} \mathbf{u}\cdot\mathbf{p}_i$
    \State $i_c\gets\arg\max_{i=1,\ldots,n} \mathbf{u}\cdot\mathbf{p}_i$
    \State $s_c\gets0$, $i_b\gets i_b-1$, $s_b\gets 1$
    \State $\mathbf{a}\gets \mathbf{p}_{i_a}$
    \Loop
    \State \textbf{if} $s_b=0$ \textbf{then} $i_b\gets i_b-1$, $s_b\gets 1$
    \State \textbf{if} $s_c=1$ \textbf{then} $i_c\gets i_c+1$, $s_c\gets 0$
    \State $\mathbf{b}\gets \mathbf{p}_{i_b}(1-s_b) + \mathbf{p}_{i_b+1}s_b$, $\mathbf{c}\gets \mathbf{p}_{i_c}(1-s_c) + \mathbf{p}_{i_c+1}s_c$
    \State $\mathbf{e}_b\gets \mathbf{p}_{i_b+1} - \mathbf{p}_{i_b}$, $\mathbf{e}_c\gets \mathbf{p}_{i_c+1} - \mathbf{p}_{i_c}$
    \If{$\mathbf{u}\cdot\mathbf{e}_b=0$}
    \State $i_b\gets i_b-1$, $s_b\gets 1$
    \State \textbf{go to top of loop}
    \EndIf
    \If{$\mathbf{u}\cdot\mathbf{e}_c=0$}
    \State $i_c\gets i_c+1$, $s_c\gets 0$
    \State \textbf{go to top of loop}
    \EndIf
    \State \textbf{if} $\mathbf{e}_b\wedge\mathbf{e}_c \le 0$ \textbf{then break out of loop}
    \State $t_q\gets$\funcname{calculate\_tq}($\mathbf{a},\mathbf{b},\mathbf{e}_b,\mathbf{c},\mathbf{e}_c$,$\mathbf{u}$)
    \State $t_b\gets-t_q/(\mathbf{u}\cdot\mathbf{e}_b)$, $t_c\gets t_q/(\mathbf{u}\cdot\mathbf{e}_c)$
    \State \textbf{if} $t_b\le0$ or $t_c\le0$ \textbf{then break out of loop}
    \If{$t_b<s_b$ and $t_c<1-s_c$}
    \State $s_b\gets s_b-t_b$, $s_c\gets s_c+t_c$
    \State \textbf{break out of loop}
    \ElsIf{$(1-t_c)t_b < s_b t_b$}
    \State $s_{b}\gets t_b - (1-s_c)t_b/t_c$
    \State $i_c \gets i_c+1$, $s_c \gets 0$
    \Else
    \State $s_{c}\gets t_c + s_b t_c/t_b$
    \State $i_b \gets i_b-1$, $s_b \gets 1$
    \EndIf
    \EndLoop
    \State\textbf{return} $i_a,i_b,s_b,i_c,s_c$
    \EndFunction
\end{algorithmic}
\end{mdframed}
\caption{An $O(n)$-time algorithm to find a triangle in an $n$-gon $P$ anchored to $\mathbf{u}$.}
\label{alg:anchored}
\end{listing}

\begin{listing}
\begin{mdframed}[backgroundcolor=codebgcolor] 
\begin{algorithmic}
    \Function{largest\_inscribed\_triangle}{$(\mathbf{p}_1,\ldots,\mathbf{p}_n)$}
    \State $i_a,i_b,s_b,i_c,s_c \gets$ \funcname{anchored\_triangle}($(\mathbf{p}_1,\ldots,\mathbf{p}_n)$,$(1,0)$)
    \State $i_a^{(\text{init})}\gets i_a$, $A_\text{max} \gets 0$
    \While{$i_a\le i_a^{(\text{init})}+n$}
    \State \textbf{if} $s_b=1$ \textbf{then} $i_b\gets i_b+1$, $s_b\gets 0$
    \State \textbf{if} $s_c=1$ \textbf{then} $i_c\gets i_c+1$, $s_c\gets 0$
    \State $\mathbf{a}\gets \mathbf{p}_{i_a}$, $\mathbf{b}\gets \mathbf{p}_{i_b}(1-s_b) + \mathbf{p}_{i_b+1}s_b$, $\mathbf{c}\gets \mathbf{p}_{i_c}(1-s_c) + \mathbf{p}_{i_c+1}s_c$
    \State $\mathbf{e}_a\gets \mathbf{p}_{i_a+1} - \mathbf{p}_{i_a}$, $\mathbf{e}_b\gets \mathbf{p}_{i_b+1} - \mathbf{p}_{i_b}$, $\mathbf{e}_c\gets \mathbf{p}_{i_c+1} - \mathbf{p}_{i_c}$
    \If{$\mathrm{area}(\mathbf{a}\mathbf{b}\mathbf{c}>A_\text{max})$ and $s_b=0$ and $s_c=0$}
    \State $A_\text{max}\gets\mathrm{area}(\mathbf{a}\mathbf{b}\mathbf{c})$, $i_a^{(\text{max})}\gets i_a$, $i_b^{(\text{max})}\gets i_b$, $i_c^{(\text{max})}\gets i_c$
    \EndIf
    \State \textbf{if} $\mathbf{e}_a\wedge(\mathbf{c}-\mathbf{b})=0$ \textbf{then} $i_a\gets i_a+1$, $\mathbf{a}\gets \mathbf{p}_{i_a}$
    \State $Q\gets$\funcname{calculate\_Q}($\mathbf{a},\mathbf{b},\mathbf{e}_b,\mathbf{c},\mathbf{e}_c$)
    \If{$Q>0$}
    \State $t_{b,1}\gets1-s_b$
    \State $t_{b,3}\gets(\mathbf{e}_a\wedge(\mathbf{c}-\mathbf{b}))/(\mathbf{e}_a\wedge\mathbf{e}_b)$
    \State $t_{b,4}\gets$\funcname{calculate\_tb4}($\mathbf{a},\mathbf{b},\mathbf{e}_b,\mathbf{c},\mathbf{e}_c$)
    \State $t_b\gets$smallest positive value among $t_{b,1}$, $t_{b,3}$, and $t_{b,4}$.
    \State $s_b\gets s_b+t_b$
    \ElsIf{$Q<0$}
    \State $t_{c,2}\gets1-s_c$
    \State $t_{c,3}\gets(\mathbf{e}_a\wedge(\mathbf{b}-\mathbf{c}))/(\mathbf{e}_a\wedge\mathbf{e}_c)$
    \State $t_{c,4}\gets$\funcname{calculate\_tc4}($\mathbf{a},\mathbf{b},\mathbf{e}_b,\mathbf{c},\mathbf{e}_c$)
    \State $t_c\gets$smallest positive value among $t_{c,2}$, $t_{c,3}$, and $t_{c,4}$.
    \State $s_c\gets s_c+t_c$
    \ElsIf{$Q=0$}
    \State $\alpha,\beta,\gamma\gets$\funcname{calculate\_alpha\_beta\_gamma}($\mathbf{a},\mathbf{b},\mathbf{e}_b,\mathbf{c},\mathbf{e}_c$)
    \State $t_{b,1}\gets1-s_b$, $t_{c,1}\gets\beta t_{b,1}/(\gamma-\alpha t_{b,1})$
    \State $t_{c,2}\gets1-s_c$, $t_{b,2}\gets\gamma t_{c,2}/(\beta+\alpha t_{c,2})$
    \State $t_{b,3},t_{c,3}\gets$\funcname{calculate\_tb3\_tc3}($\mathbf{a},\mathbf{e}_a,\mathbf{b},\mathbf{e}_b,\mathbf{c},\mathbf{e}_c$)
    \State $t_b,t_c\gets$ pair with smallest $t_b$ among the pairs $(t_{b,1},t_{c,1})$, $(t_{b,2},t_{c,2})$, and $(t_{b,3},t_{c,3})$ with $t_b\ge0$ and $t_c\ge0$.
    \State $s_b\gets s_b+t_b$, $s_c\gets s_c+t_c$
    \EndIf
    \EndWhile
    \State\textbf{return} $i_a^{(\text{max})}$, $i_b^{(\text{max})}$, $i_c^{(\text{max})}$
    \EndFunction
\end{algorithmic}
\end{mdframed}
\caption{An $O(n)$-time algorithm to find a triangle inscribed in an $n$-gon $P$ of maximal area.}
\label{alg:max_inscribed}
\end{listing}

\begin{listing}
\begin{mdframed}[backgroundcolor=codebgcolor] 
\begin{algorithmic}
    \Function{calculate\_Q}{$\mathbf{a},\mathbf{b},\mathbf{v},\mathbf{c},\mathbf{w}$}
    \State \textbf{return} $(c_2-b_2)[(b_2-c_2)v_1w_1 - (a_1-c_1)v_2w_1 + (a_1-b_1)v_1w_2] + (c_1-b_1)[(b_1-c_1)v_2w_2 - (a_2-c_2)v_1w_2 + (a_2-b_2)v_2w_1]$
    \EndFunction
    \Function{calculate\_tb4}{$\mathbf{a},\mathbf{b},\mathbf{v},\mathbf{c},\mathbf{w}$}
    \State \textbf{return} $\{v_1w_1(b_2-c_2)(b_2-c_2)+v_2w_2(b_1-c_1)(b_1-c_1)+v_1w_2[ (a_1-b_1)b_2 + (c_1-b_1)a_2 + (2b_1-a_1-c_1)c_2]+v_2w_1] (b_1-c_1)a_2 - (c_1-a_1)c_2 + (2c_1-a_1-b_1)b_2]\}/[(v_1w_2 - v_2w_1)(v_1(a_2 + b_2 - 2c_2) - v_2(a_1 + b_1 - 2c_1))]$
    \EndFunction
    \Function{calculate\_tc4}{$\mathbf{a},\mathbf{b},\mathbf{v},\mathbf{c},\mathbf{w}$}
    \State \textbf{return} $-\{v_1w_1(b_2-c_2)(b_2-c_2)+v_2w_2(b_1-c_1)(b_1-c_1)+v_1w_2( (a_1-b_1)b_2 + (c_1-b_1)a_2 + (2b_1-a_1-c_1)c_2)+v_2w_1( (b_1-c_1)a_2 - (c_1-a_1)c_2 + (2c_1-a_1-b_1)b_2)\}/[(v_1w_2 - v_2w_1)(w_1(a_2 + c_2 - 2b_2) - w_2(a_1 + c_1 - 2b_1))]$
    \EndFunction
    \Function{calculate\_alpha\_beta\_gamma}{$\mathbf{a},\mathbf{b},\mathbf{v},\mathbf{c},\mathbf{w}$}
    \State \textbf{return} $2(v_1w_2 - w_1v_2)$, $v_1(2c_2-a_2-b_2) - v_2(2c_1-b_1-a_1)$, $w_1(2b_2-a_2-c_2) - w_2(2b_1-c_1-a_1)$
    \EndFunction
    \Function{calculate\_tb3\_tc3}{$\mathbf{a},\mathbf{e},\mathbf{b},\mathbf{v},\mathbf{c},\mathbf{w}$}
    \State $t_{b,3}\gets \{v_1w_1e_2(b_2-c_2)+v_2w_2e_1(b_1-c_1)+v_1w_2(e_1(b_2-a_2) - e_2(2b_1-a_1-c_1))+v_2w_1(e_2(b_1-a_1) - e_1(2b_2-a_2-c_2))\}/[2(v_1w_2 - w_1v_2)(e_2v_1 - e_1v_2)]$
    \State $t_{c,3}\gets \{v_1w_1e_2(b_2-c_2)+v_2w_2e_1(b_1-c_1)+v_2w_1(-e_1(c_2-a_2) + e_2(2c_1-a_1-b_1))+v_1w_2(-e_2(c_1-a_1) + e_1(2c_2-a_2-b_2))\}/[2(v_1w_2 - w_1v_2)(e_2w_1 - e_1w_2)]$
    \State \textbf{return} $t_{b,3}$,$t_{c,3}$
    \EndFunction
    \Function{calculate\_tq}{$\mathbf{a},\mathbf{b},\mathbf{v},\mathbf{c},\mathbf{w}$,$\mathbf{u}$}
    \State \textbf{return} $u_1[ (b_2-c_2)v_1w_1 + (c_1-a_1)v_2w_1 + (a_1-b_1)v_1w_2] - u_2[ (b_1-c_1)v_2w_2 + (c_2-a_2)v_1w_2 + (a_2-b_2)v_2w_1]$
    \EndFunction
\end{algorithmic}
\end{mdframed}
\caption{Cumbersome formulas removed from Listings \ref{alg:anchored} and \ref{alg:max_inscribed} for neatness.}
\label{alg:util}
\end{listing}

\bibliography{inscribed}

\end{document}